\documentclass[conference,letterpaper]{IEEEtran}

\usepackage{theorem}

\usepackage{amsmath}
\usepackage{amsfonts}
\usepackage{latexsym}
\usepackage{amssymb}
\usepackage{algorithm2e}
\usepackage{graphics}
\usepackage{graphicx}
\usepackage{subfigure}
\usepackage{enumerate}

\usepackage{mathrsfs} 

\usepackage{cite} 
\usepackage{upref}

\usepackage{theorem}
\usepackage{psfrag}
\usepackage{color}
\usepackage{tikz,pgfplots}

\usepackage{geometry}
\usepackage{tikz}
\usepackage{pgfplots} 
\pgfplotsset{compat=1.16}
\usetikzlibrary{pgfplots.groupplots}

\usepackage[capitalize]{cleveref}
\usepackage[inline]{enumitem}
\usepackage{mathtools}

\IEEEoverridecommandlockouts

\newcommand{\cA}{{\cal A}} 

\newcommand{\cC}{{\cal C}} 
\newcommand{\cD}{{\cal D}}

\newcommand{\cP}{{\cal P}}

\newcommand{\cS}{{\cal S}}

\newcommand{\cV}{{\cal V}}

\newcommand{\cZ}{{\cal Z}}

\DeclareMathAlphabet{\mathbfsl}{OT1}{ppl}{b}{it} 

\newcommand{\bc}{\mathbfsl{c}} 
\newcommand{\bu}{\mathbfsl{u}}

\newcommand{\be}[1]{\begin{equation}\label{#1}}
\newcommand{\ee}{\end{equation}} 
\newcommand{\eq}[1]{(\ref{#1})}

\renewcommand{\leq}{\leqslant}
 
\renewcommand{\geq}{\geqslant}

\newcommand{\script}[1]{{\mathscr #1}}

\renewcommand{\Bbb}{\mathbb}
\newcommand{\C}{{\Bbb C}} 
\newcommand{\N}{{\Bbb N}}
\newcommand{\R}{{\Bbb R}} 

\newcommand{\F}{{\Bbb F}}
\newcommand{\Lf}{{\Bbb L}}

\newcommand{\Tref}[1]{Theo\-rem\,\ref{#1}}

\renewcommand{\Cref}[1]{Co\-ro\-lla\-ry\,\ref{#1}}

\newtheorem{thm}{Theorem\hspace{-1pt}} 
\newenvironment{theorem}
{\begin{thm}\hspace*{-1ex}{\bf.}}{\end{thm}}

\newtheorem{prop}[thm]{Proposition$\!$}

\newcommand{\deff}{\mbox{$\stackrel{\rm def}{=}$}}

\newcommand{\Span}[1]{{\left\langle {#1} \right\rangle}}

\newcommand{\sM}{\script{M}}

\newcommand{\ds}{d^{(s)}}

\theoremstyle{plain} 
\theorembodyfont{\normalfont\slshape}

\newtheorem{lem}[thm]{Lemma\hspace{-.75pt}}

\newtheorem{cor}[thm]{Corollary$\!$}
\newenvironment{corollary}{\begin{cor}\hspace*{-1ex}{\bf.}}{\end{cor}}

\newtheorem{defn}{Definition$\!$}
\newenvironment{definition}{\begin{defn}\hspace*{-1ex}{\bf.}}{\end{defn}}

\newcommand{\rr}[1]{\textcolor[rgb]{0,0,0}{#1}}
\newcommand{\rrr}[1]{\textcolor[rgb]{0,0,0}{#1}}
\newcommand{\bl}[1]{\textcolor[rgb]{0,0,0}{#1}}

\begin{document}
\title{Subspace Coding for Spatial Sensing} 



\author{%
  \IEEEauthorblockN{Hessam Mahdavifar}
  \IEEEauthorblockA{Department of Electrical and Computer Engineering\\
                    Northeastern University,
                    Boston, MA, USA
                    }
  \and
  \IEEEauthorblockN{Robin Rajam\"{a}ki and Piya Pal}
  \IEEEauthorblockA{Department of Electrical and Computer Engineering \\
                    University of California San Diego, 
                    La Jolla, CA, USA
                    }                   
                    \thanks{This work was supported in part by NSF under Grant CCF-2415440, the Center for Ubiquitous Connectivity (CUbiC) under the JUMP 2.0 program as well as by NSF 2124929, ONR N00014-19-1-2256, and DE-SC0022165.}
                   
}

\maketitle

\begin{abstract}
A subspace code is defined as a collection of subspaces of an ambient vector space, where each \rr{information-encoding} codeword is a subspace. 
\rr{T}his paper \rr{studies} a class of spatial sensing problems, notably direction of arrival (DoA) estimation  \rr{using multisensor arrays}, \rr{from} a novel subspace coding \rr{perspective}. Specifically, we demonstrate how \rr{a canonical (passive) sensing} model can be mapped into a subspace coding problem, with the sensing operation defining a unique structure for the subspace codewords. We introduce the concept of \textit{sensing subspace codes} 
follow\rr{ing} this structure, \rr{and show} how the\rr{se} codes can be controlled by \rr{judiciously designing the sensor array geometry.} 
We further present a construction of \rr{sensing subspace codes} leveraging a certain class of Golomb rulers \rr{that achieve near-optimal minimum codeword distance. These designs inspire novel noise-robust sparse array geometries achieving high angular resolution. We also prove that codes corresponding to conventional uniform linear arrays are suboptimal in this regard.} This work is the first to establish connections between subspace coding and spatial sensing, \rr{with the aim of} leverag\rr{ing} insights and methodologies \rr{in one field} to tackle challenging problems \rr{in the other}.
\end{abstract}

\section{Introduction}
\label{sec:one}
A subspace code associated with \rr{an} ambient vector space $\cV$ is \rr{defined as} a subset of the set of all subspaces of $\cV$. Subspace codes \rr{find applications} in non-coherent communications where the medium only preserves the subspace spanned by the input vectors into the system rather than the individual symbol entries of the vectors. The notion of subspace codes and utilizing them for non-coherent communication was first introduced in the seminal work by Koetter and Kschischang \cite{KK} in the context of randomized network coding, where subspace codes were employed over finite fields (i.e., $\cV = \F_q^M$, for some power of prime $q$). 
Recently, this framework \rr{was} extended to the field of real/complex numbers ($\cV = \R^M$ or $\cV = \C^M$), \rr{and} \textit{analog} subspace codes \rr{were demonstrated to} enable reliable communication over wireless networks in a non-coherent fashion \cite{soleymani2022analog}. In particular, it \rr{was} shown that the subspace error and erasure-correction capability of an analog subspace code is characterized by the \emph{minimum subspace distance}\rr{---}a variation of the chordal distance \cite{soleymani2022analog}\rr{---and} new algebraic constructions for complex subspace codes were introduced \cite{soleymani2022analog,soleymani2021new}. 
Subspace codes are also closely related to Grassmann codes \cite{agrawal2001multiple, zheng2002communication}, \rr{where} the dimension of all the subspace codewords are equal. Subspace codes extend Grassmann codes by allowing subspace codewords of different dimensions to be included in the code. 

\rr{Parallel to these developments, the classical topic of \emph{sensing} has recently experienced renewed research interest, in large part due to increasingly closer integration of sensing and communications functionalities with the goal of
efficiently utilizing spatio-temporal resources and
attaining improved spatial selectivity, resolution, identifiability, and signal-to-noise-ratio (SNR)\cite{liu2020joint,ma2020joint,sun2020mimoradar,ahmadipour2022aninformation}. 
Fully harnessing these advantages requires judiciously designing the sensor array geometry, 
which has led to a great interest in \emph{sparse arrays}. 
They offer many benefits over conventional uniform arrays, such as the ability to identify more signal sources / targets than sensors, as well as super-resolution capabilities and robustness to noise \cite{liu2017cramerrao,wang2017coarrays,sarangi2023superresolution}. These sensing tasks involve estimating parameters, such as directions-of-arrival (DoAs) of emitters / targets, {\em encoded in a subspace} of the received signal. This hints at a deep underlying connection between subspace codes and sensing beckoning to be explored.}

\rr{
The purpose of this paper is to establish first contact between the two fields of subspace coding and sensing, 
showing that they actually 
share many fundamental
objectives despite having so far evolved independently. 
Herein, we take preliminary steps towards identifying core mathematical problems of mutual interest and providing solutions to a subset of them. Firstly, we establish that the array geometry (especially of sparse arrays) plays an important role in the design of \emph{structured} subspace codes, inspiring novel constructions that achieve optimal minimum subspace distance. 
A key discovery is that the difference set \bl{(e.g., see \cite{tao2006additive})}---a key mathematical object in the study of sparse arrays, arising from the sensing model---also emerges naturally from the minimum distance of subspace codes. 
Secondly, we show that the classical problem of DoA estimation using a multisensor array} can be mapped into a subspace coding problem. \rr{Interestingly,} the sensing \rr{model} imposes a unique structure on \rr{the} subspace codewords, 
\rr{which can be controlled 
by designing the array geometry.} 
This leads to the new concept of \textit{sensing subspace codes} that can be also realized in practice by spatial sampling strategies such as antenna placement and selection. To showcase the applicability of these codes \rr{and the novel insight into DoA estimation and array design yielded by them, we consider designing} a one-dimensional sensing subspace code in $\C^M$ \rr{for} \rr{a} single \rr{source / target model. We propose a code construction} leverag\rr{ing} a special type of Golomb rulers \rr{with desirable modular arithmetic properties that achieve
near-optimal minimum subspace distance. This 
gives rise to novel 
noise-robust sparse array 
designs. 
Conversely, we show that codes based on uniform sampling---corresponding to conventional uniform linear arrays---are strictly suboptimal in this regard.} 
\rr{W}e discuss several directions for future \rr{work} with the aim to motivate further research \rr{in this hitherto unexplored} intersection of coding theory and sensing.

\textbf{Notation:}
For $N\!\in\!\N$, let $[N]$ denote $\{1,2,\dots,N\}$. Also, let $\cZ(N)\bl{\,\deff\,\{e^{j2\pi n/N}, n\!\in\![N]\!-\!1\}}$ denote the set of $N$-th roots of unity. 
The column space of a matrix $X$ is denoted by $\Span{X}$. For a vector space $\cV$ let $\cP(\cV)$ denote the set of all subspaces of $\cV$. Also, 
$\cV\!=\!\Lf^M$, where $\Lf$ can be either $\R$ or $\C$. The set of all $r$-dimensional subspaces of $\Lf^M$ is denoted by $G_{r,M}(\Lf)$, which is referred to as Grassmann space. Given $\Lf\!=\!\C$, $G_{r,M}(\C)$ can be also described as 
$G_{r,M}(\C)\,\deff\,\{\Span{Z} : Z\!\in\!\C^{M\times r}, Z^{\text{H}}Z\!=\!I_r\},$ where $I_r$ is the $r\times r$ identity matrix. The elements of $G_{r,M}(\Lf)$ are also referred to as $r$-planes. 

\textbf{Subspace coding preliminaries:}
Consider two $r$-planes $U$ and $V$. Let $U_i \in U$ and $V_i \in V$ be column vectors having unit length such that $|V_i^{\text{H}}U_i |$ is maximal, subject to the conditions $U_j^{\text{H}} U_i =0$ and $V_j^{\text{H}} V_i =0$ for all $i,j$ with $i > j \geq 1$. Then the \textit{principal} angle $\beta_i$, for $i \in [r]$, between $U$ and $V$ is defined as $\beta_i=\arccos{|V_i^{\text{H}} U_i|}$, see, e.g., \cite{barg2002bounds,roy1947note}. \rrr{Similarly to} \cite{soleymani2022analog}, we adopt a \emph{quasi-distance} metric between subspaces $U$ and $V$, refer\rrr{red} to as the subspace distance, defined as
\be{ds}
\ds(U,V)\,\deff\, \sum_{i=1}^r \sin^2(\beta_i).
\ee
\rrr{Note that} \eq{ds} is proportional to the square of the chordal distance $d_c$ \cite{conway1996packing,barg2002bounds}. \rrr{The} minimum distance \rrr{of} an analog subspace code $\cC\subset\cP(\cV)$ is defined as \cite{soleymani2022analog}
\be{dsmin-def}
\ds_{\min}(\cC)\,\deff\, \min_{U,V\in \cC, U\neq V} \ds(U,V).
\ee

\section{Subspace Coding Meets Spatial Sensing}
\label{sec:three}

\subsection{Receiver System Model for \rr{Spatial Sensing}}
\label{sec:three-A}

Consider $K$ unknowns $\theta_1, \theta_2, \cdots, \theta_K$, with $\theta_k \in [-\pi/2,\pi/2)$, representing \rr{the} angles of arrival 
\rr{of} $K$ \rr{far field} objects. A discretization is assumed where $\theta_k$'s come from a grid with $N$ points. The grid points can be
arbitrarily chosen as long as they are distinct. Here, we pick $\theta_k$'s from a grid with the $\sin(\bl{\cdot})$ of the grid points equally distanced in $[-1,1)$, i.e., \rrr{$\theta_k\in\{\arcsin(-1+2(n-1)/N)\}_{n=1}^N$.} 
The \rr{(narrowband)} receiver observes $Y = [y_{m,l}]_{M \times L} \in \mathbb{C}^{M\times L}$, where $L\geq K$, with $y_{m,l}$'s, for $m \in [M], l \in [L]$ given as follows:
\be{y-model}
y_{m,l} = \sum_{k=1}^K h_m(\theta_k)x_{k,l}+w_{m,l},
\ee
where $w_{m,l}$'s are additive noise terms, often assumed to be i.i.d. complex Gaussian, and $x_{k,l}$'s are unknown 
\rr{nuisance parameters}, interpreted as the time varying \rr{source signal amplitudes}
at the $l$-th time sample of an unknown source \rr{in} direction $\theta_k$. 
In this sensing model, $M$ represents the number of receiver antennas and 
$L$ represents the number of time samples. 
We consider \rr{a canonical \emph{passive sensing} model,} 
where 
\rr{$h(\theta) = [h_m(\theta)]_{M \times 1}$, $m \in [M]$ denotes the array manifold vector describing the relative phase shifts experienced by the antennas when a plane wave impinges on the array from direction $\theta$. Hence, we have:}\footnote{In \emph{active sensing}, the transmit array geometry and waveforms may impose additional \emph{spatio-temporal} structure on $h_m(\cdot)$. This is a topic of future research.}
\be{h-model}
h_m(\theta) = e^{j \pi d_m \sin\theta},\ \text{for}\ m \in[M].
\ee 
\rr{By \eq{h-model}, $h_m(\theta_k)$ in \eq{y-model} depends on unknown angles $\{\theta_k\}_{k=1}^K$. Nevertheless, we can influence $h_m(\cdot)$ by the choice of the receiver \emph{array geometry}, represented by the set of antenna positions $\{d_m\}_{m=1}^M$ (in units of half a wavelength of the carrier frequency).} 
\rr{For simplicity, we assume the $d_m$'s are integers between $0$ and $N-1$,} 
i.e., $d_m \in [N]\rr{-1}$, for $m \in [M]$.

The goal of the receiver is to identify the $K$-tuple $(\theta_1,\theta_2,\ldots, \theta_K)$ or, equivalently, the $K$-tuple of $\sin(\theta_k)$'s, 
\rr{given}
received signal $Y_{M \times L}$ 
\rr{and array geometry $\{d_m\}_{m=1}^M$}.

\subsection{A Subspace Coding \rr{Perspective of Spatial Sensing}}

The $K$-tuple $(\theta_1,\theta_2,\ldots, \theta_K)$ can be considered as a message $\bu \in \sM$, where $\sM$ is the message alphabet of size $N\choose K$. Furthermore, the mapping $H: \sM \rightarrow \C^{M \times K}$ is the encoder function, where an input message $\bu = (\theta_1,\theta_2,\ldots, \theta_K)$ is encoded into an $M \times K$ complex matrix $H(\bu)$ as follows:
\be{encoder-model}
H(\bu) = [h(\theta_1), h(\theta_2), \ldots, h(\theta_K)]_{M \times K}.
\ee
Note that the subtle difference 
to 
a classical coding problem is that we do not have 
full control over the encoder/code design. 
\rr{Instead,} the encoder is selected from the set of all encoders with $h_m(\bl{\cdot})$ components following \eq{h-model}, \rr{where} we have control over \rr{only} the values of the $d_m$'s. 
In other words, there is a mapping from the set of $M$-tuples $(d_1,d_2,\dots,d_M)$ to the set of possible encoders $H(\bl{\cdot})$. 

Next, note that the equation for the receiver system model, expressed in \eq{y-model}, can be rewritten as follows:
\be{Y-model}
Y_{M \times L} = H(\bu)_{M \times K} X_{K \times L} + W_{M \times L},
\ee
where $X$ and $W$ are the matrices of unknown coefficients $x_{k,l}$'s and noise terms $w_{m,l}$'s, for $k \in [K]$, $m \in [M]$, $l \in [L]$. In fact, $Y$ is the received signal, i.e., the channel output in a coding problem, and the goal is to decode, 
i.e., 
recover $\bu$ from $Y$. 

Next, we discuss how the problem can be mapped to a subspace coding problem. Note that 
matrix $X$ is 
completely unknown. Hence, in the transition from $H(\bu)$ to $H(\bu)X$ the only preserved information is the subspace spanned by the columns of $H(\bu)$. This exactly mirrors the logic behind the subspace coding problem in \cite{KK,soleymani2022analog}. We refer to the overall encoding function as a \textit{DoA sensing encoder}, 
defined below. 

\begin{definition}[\rr{DoA sensing encoder}]
\label{encoder-def}
For each choice of $M$-tuple $(d_1,d_2,\dots,d_M)$, with $d_m \in [N]$, let $H(\bl{\cdot}): \sM \rightarrow \C^{M \times K}$ be the corresponding function defined in \eq{encoder-model} with the components of $h(\bl{\cdot})$ set as in \eq{h-model}. Then the corresponding DoA sensing encoder $H : \sM \rightarrow \C^{M \times K}$ takes a message $\bu = (\theta_1,\theta_2,\ldots, \theta_K) \in \sM$ as the input and generates $H(\bu)$ as the output. Furthermore, the corresponding \textit{sensing subspace code} $\cC \subset \cP(\cV)$, where $\cV = \C^M$, is defined as follows:
\be{code-def}
\cC
= \{\Span{H(\bu)}: \bu \in \sM\}. 
\ee
\end{definition}
The receiver observes $Y$, according to \eq{Y-model}, and aims to find the message $\bu$ by invoking a decoder function $\cD:  \C^{M \times L} \rightarrow \sM$.

\noindent
{\bf Remark\,1.} It is shown in \cite{soleymani2022analog} that the subspace distance metric, defined in \eq{ds}, perfectly captures the capability of a subspace code to recover from \textit{subspace errors} and \textit{subspace erasures} as well as \rr{its robustness to} 
additive noise. Note that since $L \geq K$ \rr{is assumed}, we do not run into the issue of subspace erasure \rr{provided the $d_m$'s are chosen appropriately}. 
In fact, the decoder $\cD$ only needs to recover from the additive noise. This problem is well studied and understood in classical block coding. However, such studies can not be applied to the subspace coding domain in a straightforward fashion. 
\rrr{In contrast to prior work,} this paper aim\rr{s} at designing sensing subspace codes with the specific structure \rr{in \eq{h-model}} imposed by the sensing problem, while maximizing the minimum subspace distance.
\bl{Such} 
subspace codes 
have not been studied before.


\section{Sensing Subspace Codes: New insights and (near) Optimal Designs for Single Target}\label{sec:designs}

\rr{This section examines in further detail the single source} 
case 
$K=1$, 
\rr{where} 
the goal is to estimate a single 
\rr{DoA} 
$\theta$ 
\rr{belonging to} 
the grid of $N$ points specified in Section\,\ref{sec:three-A}. 
Beyond providing valuable insights into the structure of desirable constrained subspace codes, the single source setting is also an important problem to understand in its own right. Indeed, it finds applications in both sensing and communications, 
\cite{chiu2019active, liu2023integrated}, where several theoretical questions still remain open.

\rr{T}he problem of subspace code design in 
\rr{the} case 
\rr{$K=1$}
is reduced to packing lines. 
The problem of packing lines in $G_{1,M}(\R)$ was first studied by Shannon \cite{shannon1959probability}. 
A subtle difference between this and our problem is that spherical codes are packings in $G_{1,M}(\R)$, whereas our case study for $K=1$ is reduced to code designs in $G_{1,M}(\C)$, which is actually equivalent to $G_{2,M}(\R)$. We also study this problem under the specific structure \rr{\eq{h-model}} imposed by the DoA estimation problem. 

To further simplify the expressions, let $\alpha := e^{j \pi \sin\theta}$. Due to the specific discretization of $\theta$, i.e., $\sin\theta$ belonging to a grid of $N$ equally distanced points in $[-1,1)$, we have $\alpha \in \cZ(N)$, i.e., $\alpha^N = 1$. In particular, for $K=1$, the one-dimensional sensing subspace code $\cC$, for a certain choice of $(d_1,d_2,\dots,d_M)$ as defined in Definition\,\ref{encoder-def}, is given by:
\be{code-def}
\cC = \bigl\{\Span{(\alpha^{d_1},\alpha^{d_2},\dots,\alpha^{d_M})}: \alpha \in \cZ(N)\bigr\} . 
\ee
Here, codewords are $1$-dimensional spans of single-column matrices. However, for ease of notation, a single-column matrix is represented as a row, e.g., $(\alpha^{d_1},\alpha^{d_2},\dots,\alpha^{d_M})$. 
Note that $|\cC|=N$ and $\cC \subset G_{1,M}(\C)$. For $\alpha \in \cZ(N)$, let $\bc(\alpha)$ denote the corresponding codeword in $\cC$, i.e., 
\be{codeword-def}
\bc(\alpha) = \Span{(\alpha^{d_1},\alpha^{d_2},\dots,\alpha^{d_M})}. 
\ee
Also, the subspace distance $d^{(s)}(.,.)$ between two codewords $\bc(\alpha)$ and $\bc(\alpha')$ is given by
\be{ds-alpha}
d^{(s)}(\bc(\alpha), \bc(\alpha')) = 1 - \frac{1}{M^2} \bigg|\sum_{m=1}^M (\alpha^* \alpha')^{d_m} \bigg|^2,
\ee
where $.^*$ 
\rr{denotes} 
complex \rr{conjugation.} 
\rr{T}he minimum distance $\ds_{\min}(\cC)$, defined in \eq{dsmin-def}, can be rewritten as follows:
\be{ds-min}
\ds_{\min}(\cC) \,\deff\, \min_{\alpha,\alpha' \in \cZ(N), \alpha \neq \alpha'} d^{(s)}(\bc(\alpha), \bc(\alpha')).
\ee

 For the decoder in the case of $L=K=1$, one can consider the minimum distance decoder $\cD_{\min}$ that first computes
\be{decoder}
\hat{\alpha} =  \underset{\alpha: \alpha \in \cZ(N)}{\text{arg\,\rr{max}}} |Y^{\text{H}} \bc(\alpha)|,
\ee
and then \rr{maps} $\hat{\alpha}$ 
to the corresponding \rr{output message} $\hat{\theta}$. 
\rr{The output of the minimum distance decoder in \eq{decoder} can be shown to be equivalent of the maximum likelihood estimate of $\alpha$ under i.i.d. complex circularly symmetric Gaussian noise.} 
Next, we establish a relation between $\ds_{\min}(\cC)$ and the probability of error of $\cD_{\min}$ defined naturally as 
$$
P_e(\cD_{\min}) \,\deff\, \text{Pr}\{\hat{\theta} \neq \theta\},
$$
under an i.i.d. complex Gaussian noise model for the noise terms $w_m$'s. Let the receiver 
SNR
across each antenna be $\frac{1}{\sigma^2}$, where $\sigma^2$ can be 
regarded as 
the normalized noise power.

\begin{theorem}
\label{thm1}[\rr{Probability of error of $\mathcal{D}_{\min}$}]
For a one-dimensional sensing subspace code $\cC$ with minimum subspace distance $d_{\min} := \ds_{\min}(\cC)$, $P_e(\cD_{\min})$ is bounded as follows:
\be{Pebound}
P_e(\cD_{\min}) \leq \exp{\Big(-\frac{M}{4\sigma^2} (1-\sqrt{1-d_{\min}})^2\rr{+\ln N}\Big)}.
\ee
\end{theorem}
\begin{proof}
Let $\alpha$ correspond to the true $\theta$. For any other $\alpha' \neq \alpha$, with $\alpha' \in \cZ(N)$, utilizing \eq{ds-alpha} and the code structure specified in \eq{code-def} we have
\be{lemma1-eq1}
|\bc(\alpha)^{\text{H}} \bc(\alpha')| \leq M \sqrt{1-d_{\min}}.
\ee
Suppose that the noise is normalized, i.e., $|x|=1$ for the single scattering coefficient $x$. Then by \eq{lemma1-eq1} and noting that $|\bc(\alpha)^{\text{H}} \bc(\alpha)| = M$, the decoder $\cD_{\min}$ is guaranteed to be successful 
\rr{if} 
the \rr{absolute value} of the overall noise term \rr{$W^{\text{H}}\bc(\gamma)$} 
is less than $\frac{1}{2}(M-M \sqrt{1-d_{\min}}),\rr{\forall \gamma\in\cZ(N)}$. Also, note that the variance of \rr{this} overall noise term is $M\sigma^2$. 
\rr{Eq.~\eq{Pebound}} is then established by utilizing the cumulative distribution function of the Rayleigh distribution representing the \rr{magnitudes} 
of the \rr{$N$} overall noise term\rr{s}, followed by the union bound on the probability of the error event for all $\alpha'$.
\end{proof}

The result in \Tref{thm1} shows that improving $d_{\min}$ directly improves the upper bound on the probability of error of the minimum distance decoder. This further justifies our approach to construct sensing subspace codes with the aim to maximize the minimum subspace distance. 
\subsection{Near Optimal Code Constructions}
\label{sec:four}

In this section we aim \rrr{to} 
design 
sensing subspace codes with a minimum distance that is as large as possible by carefully choosing the underlying parameters $d_1,d_2,\dots,d_M$. 
\rrr{I}n the context of block coding, the problem of maximizing the minimum distance of a code of given size and length is one of the most classical coding theory problems and has been extensively studied in the past several decades. 

Given \eq{ds-alpha}, 
maximiz\rrr{ing} $\ds_{\min}(\cC)$ \rrr{is equivalent} to minimiz\rrr{ing} (with respect to $\{d_m\}_{m=1}^M$)
\be{minimax}
\max_{\alpha \neq 1, \alpha \in \cZ(N)} \bigg|\sum_{m=1}^M \alpha^{d_m}\bigg|^2. 
\ee
Note that $\alpha^* \alpha' \in \cZ(N)$ and $\alpha^* \alpha' \neq 1$, for $\alpha, \alpha' \in \cZ(N)$ and $\alpha \neq \alpha'$, which is why we simply replace $\alpha \alpha'$ by $\alpha$ in \eq{minimax}. \rr{Recalling that $\alpha = e^{j\pi\sin\theta}$, \eq{minimax} can actually be interpreted as the maximum value of the \emph{unweighted \bl{array} beampattern} \cite{rajamaki2024effect}.}

Next, we make a connection between the problem of minimizing the expression in \eq{minimax} and the classical problem of Golomb ruler design \cite{sidon1932satz}, in the particular regime with $M \approx \sqrt{N}$. Note that one can write
\be{minimax2}
\begin{aligned}
\bigg|\sum_{m=1}^M \alpha^{d_m}\bigg|^2 &=  \bigg(\sum_{m=1}^M \alpha^{d_m}\bigg)\bigg(\sum_{m=1}^M \alpha^{d_m}\bigg)^* 
= M + \sum_{\mathclap{i,\ell \in [M], i \neq \ell }} \alpha^{d_i - d_\ell}. 
\end{aligned}
\ee
For a set $\cA \subset \R$, 
\rr{let the difference set of distinct elements be}
$$
\cA - \cA \,\deff\,\{a - a': a,a' \in \cA, a \neq a'\}.
$$ 
Then the idea is that if we pick $\cA = \{d_1,\dots,d_M\}$ in such a way that $(\cA-\cA) \mod N$ covers almost all the elements in $[N]$ exactly once, then $\sum_{d \in \cA - \cA} \alpha^d$ would be 
close to zero when normalized by $N$ (note that $\sum_{d \in [N]} \alpha^d = 0$, since $\alpha$ is an $N$-th root of unity). As a result, the problem becomes relevant to the design of Golomb ruler. We leverage this connection to obtain a \textit{good} $\cA$, 
\rr{as we discuss} 
next.

Let us recall the Bose-Chowla construction of Golomb rulers \cite{bose1962theorems}. Let $q = p^n$ be a power of a prime $p$, and let $g$ denote a primitive element in $\F_{q^2}$. Then the $q$ integers in $\cS$ defined as
\be{BC-con}
\cS\, \deff\, \{i \in [q^2-2]: g^i - g \in \F_q\}
\ee
have distinct pairwise differences modulo $q^2-1$. Furthermore, the set $(\cS - \cS) \mod (q^2-1)$ has exactly $q(q-1)$ elements and equals the set of nonzero integers less than $q^2-1$ that are not divisible by $q+1$.  

Now, let us assume $N = q^2-1$, where $q$ is a power of a prime, and $M = q$. Then the sensing subspace code based on the Bose-Chowla construction, denoted by $\cC^{\text{(BC)}}$, is formally defined as follows. 

\begin{definition}[\rr{Bose-Chowla sensing subspace code}]
\label{CBC-def}
We fix $\cA = \{d_1,\dots,d_M\}$ to be the set $\cS$ given by the Bose-Chowla construction, specified in \eq{BC-con}. Then $\cC^{\text{(BC)}}$ of size $N$ is the sensing subspace code corresponding to $\cA$, as defined in \eq{code-def}. 
\end{definition} 

The following theorem presents the result on the minimum subspace distance of the code $\cC^{\text{(BC)}}$.

\begin{theorem}[\rr{Min distance of Bose-Chowla code}]
\label{thm-dmin}
For the sensing subspace code $\cC^{\text{(BC)}}$, defined in Definition\,\ref{CBC-def}, we have
\be{ds-min2}
\ds_{\min}(\cC^{\text{(BC)}}) > 1 - \frac{2}{M}. 
\ee
\end{theorem}

\begin{proof}
Given the expression for $\ds$ in \eq{ds-alpha} and the simplification 
\eq{minimax}, \rrr{establishing} \eq{ds-min2} 
is equivalent to show\rrr{ing}
\be{ds-min3}
\bigg|\sum_{m=1}^M \alpha^{d_m}\bigg|^2 < 2M,
\ee
for any $\alpha \in \cZ(N)$ with $\alpha \neq 1$.
Note that the choice of $\cA = \{d_1,\dots,d_M\}$ as the Bose-Chowla Golomb ruler ensures that $\cA^- := (\cA - \cA) \mod N$, where $N=q^2-1$, contains $N-M+1$ distinct elements from $[N-1] \cup \{0\}$, each of them exactly once. Also, note that since $\alpha \in \cZ(N)$ we have
\be{thm-dmin-eq1}
\sum_{i,\ell \in [M], i \neq \ell } \alpha^{d_i - d_\ell}  = \sum_{m \in \cA^-} \alpha^{m} = \qquad -
\sum_{\mathclap{m' \in [N-1] \cup \{0\} \setminus \cA^-}} \alpha^{m'}.
\ee
Therefore, we have
\be{thm-dmin-eq2}
\big|\sum_{i,\ell \in [M], i \neq \ell } \alpha^{d_i - d_\ell} \big| \leq |[N-1] \cup \{0\} \setminus \cA^-| = M-1,
\ee
where we used \eq{thm-dmin-eq1} together with the fact that $|\cA^-| = N -M+1$. This together with the expression in \eq{minimax2} imply that 
\be{ds-min4}
\bigg|\sum_{m=1}^M \alpha^{d_m}\bigg|^2 \leq 2M-1,
\ee
which, by \eq{ds-min3}, \rrr{completes the proof.}
\end{proof}

\begin{corollary}
The probability of error of the minimum distance decoder $\cD_{\min}$ for the code $\cC^{\text{(BC)}}$ is upper bounded as follows:
\be{Pebound2}
P_e(\cD_{\min}) < \exp{\Big(-\frac{M}{4\sigma^2} \Big(1-\frac{\sqrt{2}}{\sqrt{M}}\Big)^2 \rr{+\ln N} \Big)},
\ee
where $\sigma^2$ is the normalized power of a single noise term in the receiver system model.
\end{corollary}
\begin{proof}
The proof is by the lower bound on the minimum distance established in \Tref{thm-dmin} together with \Tref{thm1}.
\end{proof}

Next, we comment on the optimality of the code $\cC^{\text{(BC)}}$. Note that two lines being at subspace distance one from each other implies that the two lines are orthogonal. And the code $\cC^{\text{(BC)}}$ has $N=M^2-1$ lines that, by \Tref{thm-dmin}, are \textit{almost} orthogonal. An upper bound on the best minimum distance of a code of size $N$ in $G_{1,M}(\C)$ can be obtained by the Welch bound, that is $\approx 1 - 1/\sqrt{N}$. In this sense the \textit{gap-to-one} of $\ds_{\min}(\cC^{\text{(BC)}})$ is only within a factor of $2$ of the best one can hope to achieve, according to the Welch bound. Note that the character-polynomial (CP) code of \cite{soleymani2022analog}, when considered in the same regime, attains the distance $\approx 1 - 1/M$ almost matching the Welch bound. However, 
our sensing subspace code has the specific structure imposed by the sensing problem 
\rr{unlike the CP code in \cite{soleymani2022analog}.} 
Yet, our codes operate very close to what is best achievable by codes without specific structural constrains.

\subsection{Sub-optimality of uniform sensing}

Uniform linear arrays (ULA) have been the dominant choice of array geometry for decades in sensing and communication \cite{pesavento2023three}. 
The ULA is the spatial analogue of uniform (Nyquist) sampling, and hence its (spatial) frequency domain properties as captured by beamforming, are well-understood and amenable to analysis. Furthermore, the uniform geometry allows the development of fast array processing methods, often based on the Fast Fourier Transform (FFT) algorithm. We now take a critical look at this popular choice from the perspective of subspace coding, and evaluate if it is a desirable choice as a sensing code in the regime $N\approx M^2$. A ULA sensing code is given by the set \be{ULA-con}
\cS_{\text{ULA}}\, \deff\, [M]-1.
\ee
In other words, the sensor locations $d_m$ are consecutive integers from $0$ to $M-1$. It is easy to see that the set $\cS_{\text{ULA}}-\cS_{\text{ULA}}$ has $2M-1$ consecutive integers from $-M+1$ to $M-1$. As earlier, we consider the regime $N=M^2-1$. Let $\cC^{\text{(ULA)}}$ denote the sesnsing subspace code corresponding to $\cS_{\text{ULA}}$. The minimum distance of the code $\cC^{\text{(ULA)}}$ is upper bounded as 
\be{min-ULA}
\ds_{\min}(\cC^{\text{(ULA)}}) \leq 1-\frac{4}{\pi^2}. 
\ee
This result can be proved as follows. From \eq{minimax}, \be{ULA_dmin} \begin{aligned} \ds_{\min}&(\cC^{\text{(ULA)}}) =1-\frac{1}{M^2}\max_{\alpha \neq 1, \alpha \in \cZ(N)} \bigg|\sum_{m=0}^{M-1} \alpha^{m}\bigg|^2 \\ & \leq1 -\frac{1}{M^2} \bigg|\sum_{m=0}^{M-1} e^{j\frac{2\pi m}{N}}\bigg|^2  
= 1-\frac{1}{M^2} \bigg|\frac{\sin{\frac{\pi M}{N}}}{\sin{\frac{\pi}{N}} }\bigg|^2. 
\end{aligned}
\ee
In the regime $N=M^2-1$ and $M>3$, we have $\frac{M}{N}<\frac{1}{2}$. Using Jordan's inequality, which states that $\frac{2}{\pi}x \leq \sin{x} \leq x, \, 0\leq x\leq \frac{\pi}{2}$, we obtain that $$\frac{1}{M^2}\bigg|\frac{\sin{\frac{\pi M}{N}}}{\sin{\frac{\pi}{N}} }\bigg|^2 \geq \frac{1}{M^2}\frac{4M^2}{\pi^2} \geq \frac{4}{\pi^2}.$$ We obtain \eq{min-ULA} by using the above result in \eq{ULA_dmin}.
Hence, in the regime $N=M^2-1$ (can be relaxed to larger regime), the minimum distance of the ULA subspace code remains strictly bounded away from $1$ even as $M$ tends to infinity. Therefore, it falls short of achieving the Welch bound, and even asymptotically, these codewords do not become orthogonal. 
\vspace{-0.2cm}
\subsection{\rr{Numerical Results}}
We conclude this section by numerically illustrating the developed theoretical results. \cref{fig:dmin} shows the minimum subspace distance $\ds_{\min}$ in \eq{ds-min} of the Bose-Chowla \cite{dimitromanolakis2002analysis} sensing code \eq{BC-con} and the ULA \eq{ULA-con} sensing code, as a function of $M$, with $N=M^2-1$. The minimum distance approaches $1$ as $M$ increases in case of the Bose-Chowla ruler, whereas it is strictly less than $1$ (in fact, approaching $0$) in case of the ULA. This is consistent with \eqref{ds-min2} and \eq{min-ULA}, respectively. Hence, all codewords of the Bose-Chowla codebook are approximately orthogonal, whereas the ULA codebook contains codewords that are practically indistinguishable in presence of noise. This is validated by \cref{fig:Pe}, which shows the empirical probability of error of the minimum distance decoder \eq{decoder}, along with an upper bound on the true probability of error (the minimum of the trivial bound $P_e \leq 1$ and \eq{Pebound} evaluated using the values of $\ds_{\min}$ in \cref{fig:dmin}), both as a function of SNR (left panel, $M=19$) and $M$ (right panel $\text{SNR}=0$ dB). The probability of error of the Bose-Chowla codebook rapidly approaches $0$ at much lower SNR or values of $M$ than the ULA. 
The results are averaged over $10^5$ Monte Carlo trials assuming a unit-magnitude source $x\!=\!\frac{1+j}{\sqrt{2}}$, i.i.d. complex circularly symmetric Gaussian noise of variance $\sigma^2$ ($\text{SNR}:=-20\log_{10}(\sigma)$), and a source DoA $\theta$ picked uniformly at random from grid points $\arcsin(-1+2(n-1\bl{)}/N)$, $n\in[N]$.
\begin{figure}
    	\centering
    	\begin{tikzpicture} 
    		\begin{axis}[width=4.75 cm,height= 3 cm,ylabel={$\ds_{\min}(\mathcal{C})$},xlabel= {$M$},xmin=2,xmax=1e2, ymin =1e-3,ymax=  2,
      xlabel shift = {-12 pt},xticklabels={},extra x ticks={2,100},extra x tick labels={$2$,$100$},
      ymode=log,xmode=log,legend style = {at={(0.5,1.03)},anchor=south,draw=none,fill=none},legend columns=2,legend style={/tikz/every even column/.append style={column sep=0.2cm}}]
         
            \addplot[black,thick,draw] table[x=M,y=dmin]{M_vs_dmin_BCG_M_149.dat};
            \addlegendentry{Bose-Chowla}
			\addplot[red,dashed,thick,draw] table[x=M,y=dmin]{M_vs_dmin_ULA_M_317.dat};
            \addlegendentry{ULA}
    		\end{axis}%
    	\end{tikzpicture}\vspace{-0.3cm}
     \caption{\rr{Minimum distance of sensing subspace codes. As the code length $M$ (number of antennas) increases, $\ds_{\min}$ approaches its maximum value $1$ in case of the Bose-Chowla ruler, and minimum value $0$ in case of ULA ($N\!=\!M^2\!-1$).}}\label{fig:dmin}\vspace{-.3cm}
    \end{figure}
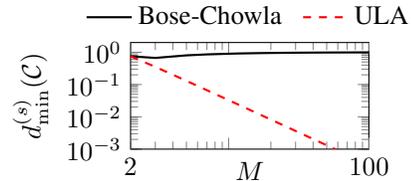
    \begin{figure}
        \centering
        \newcommand{\M}{19}
        \begin{tikzpicture}
	\begin{groupplot}[
		group style={
			group size=2 by 1,
			horizontal sep=.5cm,
			vertical sep=0cm,
		},
		width=4.75 cm,height=3 cm,ymin=1e-3,ymax=2,
  ylabel={$P_e(\mathcal{D}_{\min})$},xticklabel shift = 0 pt,yticklabel shift=0pt,ymode=log,title style={yshift=-82 pt},xlabel shift = {-12 pt},ylabel shift = {0 pt},legend style={/tikz/every even column/.append style={column sep=0.2cm}}
		]%
		\nextgroupplot[%
		,legend to name=grouplegend
		,mark=none,legend style = {draw=none,fill=none},legend columns=3,
  xlabel={SNR},xmin=-10,xmax=10,xticklabels={},extra x ticks={-10,-5,5,10}
		]
		\addlegendimage{empty legend}
         \addlegendentry{Empirical:}
             \addplot[black,thick,draw] table[x=SNR,y=Pe]{SNR_vs_Pe_BCG_SNR_10.dat};
            \addlegendentry{Bose-Chowla}
            
			\addplot[red,dashed,thick,draw] table[x=SNR,y=Pe]{SNR_vs_Pe_ULA_SNR_10.dat};
            \addlegendentry{ULA}
            
            \addlegendimage{empty legend}
            \addlegendentry{Upper bound:}
            \addplot[black,dashdotted,thick,draw,opacity=0.5] table[x=SNR,y expr={min(1,(\M^2-1)*exp(-\M*(1-sqrt(1-\thisrow{dmin}))^2/(4*10^(-\thisrow{SNR}/10))))}]{SNR_vs_Pe_BCG_SNR_10.dat};
            \addlegendentry{Bose-Chowla}
            
            \addplot[red,dotted,thick,draw,opacity=0.5] table[x=SNR,y expr={min(1,(\M^2-1)*exp(-\M*(1-sqrt(1-\thisrow{dmin}))^2/(4*10^(-\thisrow{SNR}/10))))}]{SNR_vs_Pe_ULA_SNR_10.dat};
            \addlegendentry{ULA}
  
		\nextgroupplot[ylabel={},
  xlabel={$M$},xmode=log,xmin=2,xmax=100,yticklabels={},xticklabels={},extra x ticks={2,100},extra x tick labels={$2$,$100$}]
           \addplot[black,thick,draw] table[x=M,y=Pe]{M_vs_Pe_BCG_M_149.dat};
            
			\addplot[red,dashed,thick,draw] table[x=M,y=Pe]{M_vs_Pe_ULA_M_100.dat};
            
            \addplot[black,dashdotted,thick,draw,opacity=0.5] table[x=M,y expr={min(1,(\thisrow{M}^2-1)*exp(-\thisrow{M}*(1-sqrt(1-\thisrow{dmin}))^2/(4*10^(-0/10)))}]{M_vs_Pe_BCG_M_149.dat};
            
            \addplot[red,dotted,thick,draw,opacity=0.5] table[x=M,y expr={min(1,(\thisrow{M}^2-1)*exp(-\thisrow{M}*(1-sqrt(1-\thisrow{dmin}))^2/(4*10^(-0/10))))}]{M_vs_Pe_ULA_M_100.dat};
	\end{groupplot}
	\node at (group c1r1.north) [anchor=north, yshift=1.2cm, xshift=2cm] {\ref{grouplegend}};
\end{tikzpicture}\vspace{-.3cm}
        \caption{\rr{Probability of error of minimum distance decoder. The Bose-Chowla ruler achieves a low probability of error $P_e$ both when SNR (left) or $M$ (right) increases. The ULA is significantly less robust to noise due to the smaller minimum distance of the associated sensing subspace code---indeed, when the SNR is fixed, $P_e$ approaches $1$ even as $M$ grows ($N=M^2-1$).}}
        \label{fig:Pe}\vspace{-.2cm}
    \end{figure}
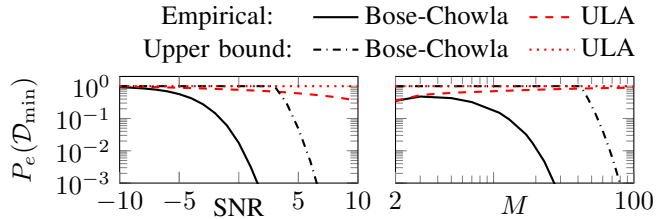

\section{Discussion and Future Directions}
\label{sec:six}
To the best of \rr{our} knowledge, this is the first work to 
\rr{unveil fundamental connections between subspace codes and sensing.} 
We hope that our work motivates further research at the intersection of \rr{these two fields,} 
\rr{which we believe constitutes a fertile ground for new ideas.} 
To this end, we \rr{briefly} discuss several directions \rr{for future} research. \rr{Firstly, sensing subspace codes with robustness to perturbations in the $d_m$'s or the $N$ grid points would be highly valuable in practical applications.} A \rr{related} question is \rr{whether} we \rr{can} extend \rr{these code designs beyond the considered regime $N\approx M^2$}? Such a question will naturally lead to investigation of higher-order difference sets. 
Another immediate extension is higher-dimensional sensing subspace codes \rr{where} 
$K > 1$. \rr{This} in general is a very challenging problem, \rr{with the only currently available results being} asymptotic limits on packing in Grassmann manifolds (as $M$ grows large, and $K$ is fixed) \cite{barg2002bounds}, and \rr{a} few explicit analog subspace code constructions beyond $K = 1,2$ \cite{soleymani2021new}. 
\rr{Further adding to the rich structure of $h(\bl{\cdot})$ are constraints imposed by} low-complexity hybrid \rr{beamforming} architectures 
\cite{koochakzadeh2020compressed}, \rr{or in the case of active sensing, the transmitted waveforms \cite{rajamaki2023importance}.} \rr{These can generally be modeled by} a \rr{structured} compression 
\rr{(wide)} matrix $A$, \rr{such that the effective system matrix in \eq{Y-model} is} $AH$. 
This problem has connections to subspace erasure problems and results from coding theory counterparts can provide new designs and guarantees. 

\bibliographystyle{IEEEtran}
\bibliography{references}



\end{document}